\documentclass[sigconf, nonacm]{acmart}

\usepackage{array}
\usepackage{booktabs}
\usepackage{amsmath}
\usepackage{pifont}
\usepackage{xcolor}
\usepackage{tikz}
\usepackage{ifthen}
\usepackage{algorithm}
\usepackage{algpseudocode}
\usetikzlibrary{arrows.meta, positioning, fit, backgrounds, calc}

\renewcommand\footnotetextcopyrightpermission[1]{}
\newcommand{\sig}{\textsuperscript{$\ast$}}
\newcommand{\yes}{\ding{51}}
\newcommand{\no}{--}

\newcommand{\pTenS}{p10s $\cdot$ req}
\newcommand{\logTSmr}{\mbox{logTS-MR} $\cdot$ item}
\newcommand{\logTStr}{\mbox{logTS-TR} $\cdot$ item}
\newcommand{\bucketTS}{\mbox{bucketTS-p75} $\cdot$ item}
\newcommand{\extshareUser}{extShare $\cdot$ viewer}
\newcommand{\seemore}{seeMore $\cdot$ req}

\definecolor{cPred}{HTML}{1F4E9C}   %
\definecolor{cBase}{HTML}{B35300}   %
\definecolor{cGate}{HTML}{228B80}   %
\newcommand{\roleP}[1]{\textcolor{cPred}{#1}}
\newcommand{\roleB}[1]{\textcolor{cBase}{#1}}
\newcommand{\roleG}[1]{\textcolor{cGate}{#1}}
\newcommand{\swatch}[1]{\textcolor{#1}{\rule{0.7em}{0.7em}}}

\newcommand{\gap}[1]{\textcolor{red}{\textbf{[GAP: #1]}}}

\newcommand{\nJD}{11}
\newcommand{\nPairs}{13}
\newcommand{\nJDSig}{7}

\newcommand{\nJDCandsAll}{4}
\newcommand{\nCands}{6}
\newcommand{\nLegs}{26}
\newcommand{\nLegsPos}{17}
\newcommand{\nLegsNeg}{1}
\newcommand{\nLegsSes}{10}
\newcommand{\nLegsSesPos}{5}
\newcommand{\nLegsMech}{16}
\newcommand{\nLegsMechPos}{12}
\newcommand{\nSesPairs}{5}
\newcommand{\nSesJD}{4}

\newcommand{\nDurHidden}{6}
\newcommand{\extShareZero}{89}

\newcommand{\meanAtPct}{82}
\newcommand{\cvRealRatio}{68}

\newcommand{\capOverview}{%
Overview of our methodology. \textbf{(a)} Retrieval and early- and late-stage ranking
each score candidates with a value model (VM), a weighted sum of predictions,
each prediction mixing a \emph{baseline} $b$ with the candidate's \emph{incremental}
excess $\Delta$. We keep the raw prediction and give $\Delta$ its own weight,
with no new model. \textbf{(b)} Viewer and item baselines are precomputed
historical rates (the row and column averages of the grid); the request
baseline is read from the live candidate pool. \textbf{(c)} One online test
ships one point; a separate four-arm ablation tests the two-term form.}

\newcommand{\capBatch}{%
[Offline] \textbf{(a)} Within-request distributions, each normalised to its
own p10--p90 range. \textbf{(b)} Coefficient of variation of each baseline
statistic across simulated requests. The simulation reproduces the ordering the
rule needs but understates absolute instability: real requests vary up to
\cvRealRatio$\times$ more.}

\newcommand{\capOffon}{%
[Offline] Every audited signal on its bias axis. Filled markers carry a term
tested online. Their outcomes are in Table~\ref{tab:main}, except
for item-axis profile tap, whose term is still being iterated. The bucketised
labels (ablated together) are an intervention's \emph{output}, so the video
label's \mbox{$r\approx0$} reflects that intervention working, not a low
reading.}

\begin{document}

\title{Baseline and Incremental: User, Item, and Request Debiasing in
Instagram Feed Ranking}

\author{Hsin-Chiang~Chen\texorpdfstring{$^{\ast}$}{}, Hao~Wang\texorpdfstring{$^{\ast}$}{}, Gedi~Zhou\texorpdfstring{$^{\dagger}$}{}, Mingze~Gao, Abhishek~Kumar, Yiben~Yang, Haoyu~Wang, Hongjing~Xia, Bo~Chen, Zhengyu~Su, Xian~Sun, Chang~Xu, Kevin~Han\texorpdfstring{$^{\ddagger}$}{}, Weiwei~Li, Zhongchang~Sun, Yao~Zhang, Qifei~Wang, Peggy~Yang, Yihan~Wu, Haoran~Liu, Siyang~Sun, Xingyu~Pan\texorpdfstring{$^{\ddagger}$}{}, Shannon~Phu, Haoming~Chen, Kevin~Zhao, Chen~Liang}
\affiliation{\institution{Meta}\country{USA}}
\email{\{hsinchiangchen, wangh, gedi, gaomingze, abi, yibenyang, hwang420, hjx, chenbo, zhengyu, xiansun, changxu109, kevinwh, weiweili90, zsun16, yaozhang, qifeiw, peggyy, yihanwu, ryan0814, siyangsun, xingyupan, shannonphu, haomingc, kevinzhao, cliang7\}@meta.com}

\begin{abstract}
A feed ranker scores candidates with a \emph{value model}, a weighted sum of
engagement predictions. Each prediction conflates a \emph{baseline}, what this
viewer, item or request would produce anyway, with the \emph{incremental}
amount by which this candidate exceeds it. The literature mostly removes the
baseline. We argue that a multi-signal value model should keep the raw
prediction and add the incremental part as a separately weighted term. On
Instagram Feed, three debiasing families deployed independently, for
viewer habit, item popularity and request context, are one parameterised
operator, and a design procedure picks a point in it. Media-type (format) bias
goes instead to a debiased label chosen by the same procedure. In a controlled four-arm
ablation of deployed Feed terms (neither term, each alone, both), the
two-term arm beats both single-term arms at production weights on point
estimates in \nJD{} of \nPairs{} judged (term, metric) pairs. Of the \nLegs{}
direct contrasts, \nLegsPos{} are significantly for it and only \nLegsNeg{}
against it, on impressions. The ablation
also suggests that the offline diagnostic should screen signals, not rank
them. Such
terms have served Instagram Feed production for most of a year. Beyond Feed,
Meta, whose apps serve billions of people, has adopted the operator across
surfaces, apps and both ranking stages. Its launches there drive sessions and
daily active users, the core visitation and retention metrics.
\end{abstract}



\begin{teaserfigure}
\Description{Three panels: the value model in retrieval and in early- and late-stage ranking with its raw and incremental terms; where the viewer, item and request baselines come from; the design procedure and the four-arm ablation.}
\centering
\begin{tikzpicture}[x=1cm, y=1cm, font=\scriptsize]
\tikzset{
  pane/.style   = {draw=black!22, rounded corners=2pt, fill=black!2},
  ttl/.style    = {font=\small\bfseries, anchor=north west},
  lab/.style    = {font=\scriptsize, anchor=north west, align=left,
                   inner sep=0pt},
  bx/.style     = {draw=black!45, rounded corners=1.5pt, fill=white,
                   font=\scriptsize, align=center, inner sep=1.5pt,
                   minimum width=2.30cm, minimum height=0.50cm},
  ar/.style     = {-{Stealth[length=1.5mm]}, black!55, line width=0.4pt},
  stg/.style    = {draw=black!45, rounded corners=1.2pt, fill=white,
                   font=\scriptsize, minimum width=2.00cm,
                   minimum height=0.36cm, inner sep=1pt},
  vm/.style     = {draw=blue!55!black, rounded corners=1.2pt, fill=blue!8,
                   font=\scriptsize\bfseries, text=blue!55!black,
                   minimum height=0.36cm, inner sep=1pt},
}
\colorlet{cbase}{black!40}
\colorlet{cincr}{blue!58!black}
\def\paneW{5.70}
\def\paneH{3.68}
\def\gap{0.25}
\newcommand{\pane}[2]{%
  \node[pane, minimum width=\paneW cm, minimum height=\paneH cm,
        anchor=north west] at (#1,0.30) {};
  \node[ttl] at ({#1+0.14},0.21) {#2};
}
\pane{0}{(a) the value model: raw + incremental}
\begin{scope}[shift={(0.14,0)}]
  \node[stg, anchor=north west] (f1) at (0.00,-0.40) {retrieval};
  \node[stg, anchor=north west] (f2) at (0.00,-0.92) {early-stage ranking};
  \node[stg, anchor=north west] (f3) at (0.00,-1.44) {late-stage ranking};
  \node[vm, anchor=west] at ([xshift=1mm]f1.east) {VM};
  \node[vm, anchor=west] at ([xshift=1mm]f2.east) {VM};
  \node[vm, anchor=west] at ([xshift=1mm]f3.east) {VM};
  \draw[ar] (f1.south) -- (f2.north);
  \draw[ar] (f2.south) -- (f3.north);
  \draw[ar] (f3.south) -- ++(0,-0.24)
    node[black!65, anchor=north, inner sep=1pt] {feed};
  \node[lab, black!60, font=\scriptsize\itshape] at (2.95,-0.40)
    {item-axis example};
  \def\base{-1.58}
  \foreach \i/\lab/\b/\inc in {0/{A}/0.70/0.14, 1/{B}/0.22/0.50} {
    \pgfmathsetmacro{\bx}{2.98 + \i*0.56}
    \fill[cbase] (\bx,\base) rectangle (\bx+0.38,{\base+\b});
    \fill[cincr] (\bx,{\base+\b}) rectangle (\bx+0.38,{\base+\b+\inc});
    \node[black!65, anchor=north, inner sep=1pt]
      at (\bx+0.19,\base-0.02) {\lab};
  }
  \draw[black!40] (2.90,\base) -- (4.00,\base);
  \node[lab, cincr] at (4.12,-0.76) {incremental $\Delta$};
  \node[lab, cbase] at (4.12,-1.18) {baseline $b$};
  \node[lab, black!70] at (2.95,-1.86)
    {one weight: A always wins\\two weights: the tuner picks};
  \node[anchor=north] at ({(\paneW-0.28)/2},-2.36)
    {$s(c) = \sum_j \big[\, \underbrace{w_j^{\text{r}} p_j(c)}_{\text{raw}}
       \;+\; \underbrace{w_j^{\text{i}} \Delta_j(c)}_{\text{incremental}} \,\big]$};
\end{scope}
\pane{\paneW+\gap}{(b) where the baseline comes from}
\begin{scope}[shift={(\paneW+\gap+0.14,-0.28)}]
  \def\cs{0.32}
  \foreach \r in {0,...,4} {
    \foreach \c in {0,...,5} {
      \pgfmathsetmacro{\v}{20 + 11*mod(\r*3+\c*2,7)}
      \fill[blue!\v!white] ({0.30+\c*\cs},{-0.50-\r*\cs})
                  rectangle ({0.30+\c*\cs+\cs-0.03},{-0.50-\r*\cs-\cs+0.03});
    }
  }
  \node[black!60, anchor=south] at ({0.30+3*\cs},-0.48) {items};
  \node[black!60, rotate=90, anchor=south] at (0.26,{-0.50-2.5*\cs}) {viewers};
  \foreach \r in {0,...,4} {
    \pgfmathsetmacro{\w}{0.12 + 0.05*mod(\r*3+1,5)}
    \fill[teal!62!black] ({0.30+6*\cs+0.05},{-0.50-\r*\cs-0.03})
                rectangle ({0.30+6*\cs+0.05+\w},{-0.50-\r*\cs-\cs+0.03});
  }
  \foreach \c in {0,...,5} {
    \pgfmathsetmacro{\h}{0.10 + 0.05*mod(\c*2+3,5)}
    \fill[orange!70!black] ({0.30+\c*\cs},{-0.50-5*\cs-0.06})
                rectangle ({0.30+\c*\cs+\cs-0.03},{-0.50-5*\cs-0.06-\h});
  }
  \draw[red!70!black, line width=0.7pt, rounded corners=1pt]
    (0.28,{-0.50-2*\cs+0.01}) rectangle ({0.30+6*\cs-0.01},{-0.50-3*\cs+0.01});
  \def\lx{2.95}
  \node[lab, teal!50!black] at (\lx,-0.52)
    {\textbf{viewer}: historical rate,\\precomputed};
  \node[lab, red!60!black] at (\lx,{-0.50-2*\cs+0.02})
    {\textbf{request}: p75 or mean\\of the live candidate pool};
  \node[lab, orange!55!black] at (\lx,{-0.50-5*\cs-0.06})
    {\textbf{item}: historical rate or\\size bucket, precomputed};
\end{scope}
\pane{2*\paneW+2*\gap}{(c) ship a term; ablate the design}
\begin{scope}[shift={(2*\paneW+2*\gap+0.14,0)}]
  \node[bx, anchor=north west] (d1) at (0.00,-0.40)
    {is there a baseline\\component?};
  \node[bx, anchor=north west] (d2) at (0.00,-1.06)
    {which baseline,\\which form?};
  \node[bx, anchor=north west] (d3) at (0.00,-1.72)
    {VM gate or\\debiased label?};
  \node[bx, anchor=north west, draw=blue!55!black, fill=blue!5] (d4)
    at (0.00,-2.38) {\textbf{ship test}: add $\Delta$\\on top of production};
  \draw[ar] (d1.south) -- (d2.north);
  \draw[ar] (d2.south) -- (d3.north);
  \draw[ar] (d3.south) -- (d4.north);
  \def\gs{0.60}\def\gcx{4.05}
  \pgfmathsetmacro{\gx}{\gcx-\gs}
  \def\gy{-1.02}
  \node[black!60, font=\scriptsize\itshape, anchor=north, align=center]
    at (\gcx,-0.36) {ablation test for understanding};
  \node[black!60, anchor=south] at (\gcx,\gy+0.02) {$w^{\text{r}}\!\rightarrow$};
  \node[black!60, rotate=90, anchor=south] at (\gx-0.04,\gy-\gs)
    {$\leftarrow\! w^{\text{i}}$};
  \foreach \i/\j/\lab/\f in {0/0/{none}/{black!8}, 1/0/{raw}/{black!22},
                             0/1/{incr}/{blue!22}, 1/1/{both}/{blue!48}} {
    \fill[\f] ({\gx+\i*\gs},{\gy-\j*\gs})
      rectangle ({\gx+\i*\gs+\gs-0.05},{\gy-\j*\gs-\gs+0.05});
    \node[black!75]
      at ({\gx+\i*\gs+\gs/2-0.025},{\gy-\j*\gs-\gs/2+0.025}) {\lab};
  }
  \draw[red!70!black, line width=0.8pt, rounded corners=1pt]
    ({\gx+\gs-0.02},{\gy-\gs-0.02})
      rectangle ({\gx+2*\gs-0.03},{\gy-2*\gs+0.07});
  \node[red!60!black, anchor=north, align=center]
    at (\gcx-0.03,{\gy-2*\gs-0.06}) {does this interior point\\beat raw and incr?};
\end{scope}
\end{tikzpicture}
\caption{\capOverview}
\label{fig:overview}
\end{teaserfigure}

\let\savednolinkurl\nolinkurl
\def\nolinkurl#1{#1}
\maketitle
\let\nolinkurl\savednolinkurl
\renewcommand{\shortauthors}{Chen et al.}
{\renewcommand\thefootnote{\fnsymbol{footnote}}%
\footnotetext[1]{Both authors contributed equally to this research.}%
\footnotetext[2]{Corresponding author.}%
\footnotetext[3]{This work was performed while these authors were at Meta.}%
}

\section{Introduction}

Instagram Feed scores each candidate $c$ with a \emph{value model} (VM), a
weighted sum
\begin{equation}
s(c) \;=\; \sum_{j} w_j \, p_j(c)
\label{eq:vm}
\end{equation}
of engagement predictions $p_j$. Its weights are tuned against a
few north-star metrics by online Bayesian
optimisation~\cite{letham2019bayesian,feng2025experimenting}. The deployed VM
also multiplies some groups of terms, a structure we set aside here.
Tuned combinations of such predictions are standard at
scale~\cite{zhao2019recommending,borisyuk2024lirank}, and each $p_j$ comes
from a large learned model~\cite{covington2016deep,naumov2019deep}. Value models of this kind
decide what people see on many major social and content platforms. How
they weigh their signals is a central problem of Web-scale recommendation.

Each $p_j$ conflates a \emph{baseline} component, the level this viewer, item
or request would produce anyway, with an \emph{incremental} component,
this candidate's excess over it. On Instagram Feed, the like prediction has a
rank correlation of $0.69$ with the viewer's own historical like rate
(Table~\ref{tab:diaguser}): much of it tracks the viewer's habit, not the
candidate. In a sum of many signals, that habit gives like extra weight on
every candidate this viewer sees, whatever the candidate adds. We say
\emph{viewer} for the user side throughout.

The literature mostly \emph{removes} the baseline, and the methods that keep
part of it lack a general account of why (\S\ref{sec:related}). We argue that
removal is the wrong default here: in a four-arm ablation of \nCands{} deployed
terms, neither the raw nor the incremental term is redundant at production
weights. We show that
the three debiasing families deployed on Feed (viewer, item and request
baselines) are one parameterised operator with four design choices.
Figure~\ref{fig:overview} summarises the argument. Instagram Feed is our case
study throughout.

\smallskip\noindent\textbf{Contributions.} \emph{1. The baseline is not a defect
but a second signal that needs its own weight} (\S\ref{sec:formulation}). A habitual
liker's likes are real engagement, and a broadly loved post deserves reach.
But one coefficient $w_j$ cannot trade the two components off, and the right
trade-off is unknown a priori and differs by signal. We keep the raw
prediction and weight the incremental part separately:
\begin{equation}
s(c) \;=\; \sum_j \Big[\, \underbrace{w_j^{\text{r}}\, p_j(c)}_{\text{raw}}
\;+\; \underbrace{w_j^{\text{i}}\, \Delta_j(c)}_{\text{incremental}} \Big].
\label{eq:twoterm}
\end{equation}

\smallskip\noindent\emph{2. One operator, three baselines}
(\S\ref{sec:formulation}). The viewer and item baselines are historical
averages, one per viewer and one per item; the request baseline is read from
the live candidate pool. They measure different things, so the three coexist
in one VM. Published methods sit along the operator's axes
(Table~\ref{tab:priorpoints}).

\smallskip\noindent\emph{3. A design procedure} (\S\ref{sec:approach}) that
turns a computationally cheap offline diagnostic into a build decision. It settles whether a
signal needs the split, which baseline and functional form to use, whether to
intervene at scoring or label time, and what the term costs to serve.

\smallskip\noindent\emph{4. A controlled ablation}
(\S\ref{sec:deployment}--\S\ref{sec:setup}). We report the Feed deployment
record, then four-arm live experiments that isolate the raw term, the incremental
term and both, on \nCands{} deployed terms. They confirm two preconditions of
the offline diagnostic (\S\ref{sec:discussion}). Beyond Feed, the operator has been
adopted across Meta's ranking surfaces (\S\ref{sec:more}).

\paragraph{Setting.}
\label{sec:background}
Feed mixes posts from followed accounts (\emph{connected} content) with
recommendations. Retrieval and early- and late-stage ranking all score
candidates with value models (Figure~\ref{fig:overview}(a)). We study
late-stage ranking, which scores a few hundred candidates per request, and two
of its value models. MainRank orders
the bulk of the feed and TopRank fills the top of a session.
A multi-task model predicts each signal as a probability or as a
regression target such as expected time spent. The operator branches only on
this distinction. The VM is a formula, not a learned function, so a new term is
a configuration change that ships in hours and reverts as fast. A scoring-time
intervention therefore costs far less to build and test than a modelling one. Every term is
evaluated in a live randomised experiment. Because the weights are tuned
jointly, a new term perturbs the operating point of every other.

\section{Related Work}
\label{sec:related}

\paragraph{Removing the baseline.} Most debiasing work treats the baseline as a
confounder to remove. Inverse propensity weighting corrects selection bias in
recommender feedback~\cite{schnabel2016recommendations} and examination bias in
learning-to-rank~\cite{joachims2017unbiased,ai2018unbiased,agarwal2019estimating}.
Logit adjustment subtracts a scaled log class prior~\cite{menon2021long}, and
causal methods deconfound popularity and related biases~\cite{chen2023bias}.
Industrial rankers remove the baseline at serving time. PAL and YouTube's
shallow tower drop their learned position part~\cite{guo2019pal,zhao2019recommending},
LAFB divides each rating prediction by its mean given viewer--item
familiarity~\cite{ren2026learning}, and AlignPxtr replaces each by
its bias-conditional deviation or quantile~\cite{lin2025alignpxtr}. Label-side
methods, DVR and RAD, instead train on duration- or cohort-relative watch-time
targets~\cite{zheng2022dvr,relativeadvantage2025}.

\paragraph{Keeping part of it.} PDA multiplies a forecast popularity factor
back in at inference~\cite{zhang2021causal}, TIDE keeps a benign quality
component~\cite{zhao2022popularity}, and DICE sums interest and conformity
scores~\cite{zheng2021disentangling}. Others keep the raw score beside a
weighted second term. MACR subtracts a scaled learned bias~\cite{wei2021model},
PPAC adds weighted popularity terms~\cite{ning2024debiasing}, and DecRS and
CD2AN blend biased and debiased scores
convexly~\cite{wang2021deconfounded,chen2022cotraining}. Collaborative filtering
has long carried additive viewer and item
baselines~\cite{koren2008factorization,koren2009matrix}. In each, the kept term
is a device of the particular model. None gives a general account of why a raw
term should stay beside an incremental one, or explores the design space of
\S\ref{sec:formulation} (Table~\ref{tab:priorpoints}). DecRS and PPAC compare
the combined score with its parts
offline~\cite{wang2021deconfounded,ning2024debiasing}, CD2AN deployed its two
endpoints separately online~\cite{chen2022cotraining}, and the four arms here
are tested on live traffic. Closest are MBD~\cite{li2026mbd} and, on YouTube
Shorts, FLVM~\cite{chang2026flvm}. MBD learns cohort
baselines in a ranking-model branch and can add a hinge on the standardised gap
to the VM score. FLVM learns a baseline branch over duration, viewer propensity
and session context; the two paths are complementary offline, but only the
latent score is served. Both are model-based and adopted by training a ranking
model. The scoring-time terms here are instead VM features that go to test as
configuration changes, at essentially no capacity cost.

\begin{table}[!htb]
\caption{Published methods in the design space of Eq.~\ref{eq:operator}. Raw at serving is \emph{dropped} if the bias factor is removed or its features set to missing, and \emph{replaced} if a debiased target or score stands in.}
\label{tab:priorpoints}
\footnotesize
\setlength{\tabcolsep}{3pt}
\begin{tabular}{@{}llll@{}}
\toprule
Method & baseline $b$ & scale $d$ & raw at serving \\
\midrule
Logit adj.~\cite{menon2021long}      & class prior        & logit $-$ log prior & yes, $-\tau\log\pi$ \\
PDA~\cite{zhang2021causal}           & item popularity    & mult.\ (power)   & re-injected \\
MACR~\cite{wei2021model}             & viewer, item branch & raw gap         & yes, $-c\,b$ \\
PPAC~\cite{ning2024debiasing}        & pers./global pop.  & additive         & yes, $+\gamma, \beta$ \\
PAL~\cite{guo2019pal}                & position           & mult.\ factor    & dropped \\
Shallow tower~\cite{zhao2019recommending} & position      & additive logit   & dropped \\
D2Q~\cite{zhan2022deconfounding}     & duration group     & group CDF        & re-mapped \\
DML~\cite{zhang2023leveraging}       & dur./viewer/item group & percentile labels & replaced \\
LAFB~\cite{ren2026learning}          & familiarity        & ratio            & replaced \\
FLVM~\cite{chang2026flvm}            & dur./viewer/session & additive logit  & replaced \\
MBD~\cite{li2026mbd}                & learned cohort     & std.\ gap        & yes; $+$ in one form \\
\midrule
\textbf{This work} & viewer/item/request & logit or raw & yes, two weights \\
\bottomrule
\end{tabular}
\end{table}

\paragraph{Labels, lists and evaluation.} The deployed bucketised labels follow
DML's duration-group long-view label~\cite{zhang2023leveraging}, a thresholded
form of D2Q's within-duration quantile~\cite{zhan2022deconfounding}. We add a
rule for when such a label is worth its cost over a scoring-time gate. Scoring a
candidate against its own result list is familiar from query-level feature
normalisation~\cite{qin2010letor}, Yahoo's list-statistic
re-ranker~\cite{yin2016ranking} and learned listwise
re-rankers~\cite{ai2018learning,pei2019personalized}. Our request baseline is
one statistic of the same prediction, computationally cheap enough to sit in the VM as a
separately weighted term. Uplift and causal-embedding methods target the causal effect of a
recommendation or outcomes under random exposure~\cite{sato2020unbiased,bonner2018causal}; our incremental term
is a lighter-weight observational relative that makes no causal claim. Whether offline
evaluation predicts online outcomes is a long-standing
concern~\cite{covington2016deep,krauth2020offline,gruson2019offline,jannach2020escaping},
and effects read at launch are biased upward by
the selection that launched them~\cite{lee2018winner}. Sections~\ref{sec:setup}--\ref{sec:discussion}
add a per-term view: each term and its raw counterpart are ablated online.

\section{One Operator, Three Baselines}
\label{sec:formulation}

Fix a signal with prediction $p$ and a \emph{baseline} $b$ estimated from the
viewer, the item or the request. The incremental term of Eq.~\ref{eq:twoterm}
is
\begin{equation}
\Delta \;=\; m(p, b)\,\Big[\, \sigma\!\big(k \cdot d(p, b)\big) \;-\; \text{shift} \,\Big],
\qquad \sigma(x) = \tfrac{1}{1+\mathrm{e}^{-x}} ,
\label{eq:operator}
\end{equation}
where the magnitude $m(p, b)$ is the prediction itself, the gap $e = p - b$,
or a constant. Every deployed scoring-time term we discuss is an instance of
it, with the comparison reversed for demotions. Four
choices fix an instance: the baseline $b$, the comparison scale $d$, the
magnitude $m$, and the sharpness and centring $(k, \text{shift})$.
Eq.~\ref{eq:twoterm} places $\Delta$ additively. Where a term was also tried
multiplicatively on the same signal, gating a group of VM terms, the trade-off
differed. Placement is a fifth degree of freedom, which we do not study.

\paragraph{The baseline $b$.} A \emph{viewer} baseline is the viewer's
historical rate for this action. An \emph{item} baseline is the item's rate
across viewers (\emph{popularity} bias) or a structural attribute that
mechanically inflates the signal (\emph{media-type} bias). That attribute is
the format and, within it, its size, such as video duration. The two biases
are routed differently in \S\ref{sec:stage}. A \emph{request} baseline is a
statistic of the predictions over the request's candidates, its \emph{batch}.

\paragraph{The comparison scale $d$.} For probabilities
$d = \operatorname{logit}(p) - \operatorname{logit}(b)$. A regression signal
is compared in its own unit, $d = \hat y - b$: the time-spent gates compare
predicted with average seconds. The choice matters only at finite $k$
(\S\ref{sec:thm-scale}). We use the logit because most per-viewer rates are
tiny and log-odds resolve that tail. A second reason is that base rates differ
widely across signals, while a unitless log-odds gap is
comparable across signals under one $k$.

\paragraph{The magnitude $m$.} With $e = p - b$, promotion terms use
$m = p$, which keeps rewarding a strong candidate in a strong reference class,
or $m = e$, which is purely relative and zero at the baseline. Demotion gates
cannot use $m = p$, which vanishes at the floor, exactly where the penalty
should bite. They use a constant magnitude with the comparison reversed,
$-C\,\sigma\big(k \cdot d(b, p)\big)$, whose factor grows the further the
candidate sits below $b$. A demotion hinge uses $m = e$, i.e.\ $\min(0, e)$.
Missing this asymmetry is a common error when a promotion term is
adapted into a demotion.

\paragraph{$k$ and shift.} With a positive magnitude, $\text{shift}=0.5$ gives
a signed term that is zero at the baseline, and $\text{shift}=0$ a one-sided boost. For $k$, test $k=1$
and the hard-threshold limit $k\to\infty$ first. Where an intermediate slope
was tried against $k=1$, it did not win, and every deployed term sits at one of
the two. By Theorem~\ref{thm:gate} they are two settings of the one dial $k$;
the deployed hinge and gate also differ in $m$.

\paragraph{Hinge and gate.} The deployed linear hinges $w\max(0, e)$ sit at
\mbox{$(m, k, \text{shift}) = (e, \infty, 0)$}. The deployed sigmoid gates
$w\,m\,\sigma(k d)$ sit at $k = 1$, $\text{shift} = 0$, with $m$ the prediction
itself: $p$, or a regression prediction $\hat y$ in its own unit. Hinge and
gate differ in \emph{two} coordinates, $m$ and $k$. The hinge is the limit
of the \emph{same-magnitude} gate $e\,\sigma(ke)$, not of the deployed one
(Theorem~\ref{thm:gate}). Hinge and gate terms are thus treated as one family
in \S\ref{sec:jd}, but $m$ and $k$ are separate choices: the two deployed
time-spent gates share $b$, $d$ and $k$ and differ in $m$.

\paragraph{Deployed instances.} Typical deployed scoring-time forms are
\[
\Delta = p\,\sigma\bigl(\operatorname{logit} p - \operatorname{logit} b\bigr),
\quad
\Delta = m\,\sigma\bigl(\hat y - b\bigr),
\quad
\Delta = \max(0,\, p - b).
\]
All have $\text{shift} = 0$. The first gates a probability on the log-odds gap
($m = p$, $k = 1$); the second, a regression output $p$ mapped back to its own
unit as $\hat y$ ($d = \hat y - b$, $k = 1$, $m = p$ or $\hat y$); the third is
a hinge ($m = e$, $k \to \infty$). The studied terms take these forms
(\S\ref{sec:deployment}), except one that is a debiased label.

\subsection{Why two terms and three baselines}
\label{sec:whyraw}

\paragraph{How a baseline acts.} An item baseline varies across a request's
candidates and reorders them on the signal itself. A viewer or request baseline
is constant within a request, so $\Delta$ keeps the order of $p$ among the
candidates it applies to. It acts by
reweighting the signal against the VM's other signals, request by request, and
by crediting candidates above $b$ more than those below.

\paragraph{Why the raw term stays.} The comparison inside $\Delta$ ignores the
common level of $p$ and $b$: $\sigma$ is bounded and reads only where a
candidate sits relative to $b$. The level of $p$ thus reaches $\Delta$ only
through $m$, and not at all at $m=e$, which sees only the gap $p-b$. Ranking on $\Delta$ alone can therefore fail on
each axis. A heavy liker's genuine like, modestly above a high baseline,
earns only part of its weight against the viewer's other signals. A fluke watch on a short clip outranks a
good long video that merely meets its duration baseline. The best candidates of
a weak and a strong request are gated alike. The raw term is where the level
enters unconditionally. With two independent weights, the tuner can sit anywhere
between the purely raw and purely incremental endpoints rather than commit to
one. This is the a-priori argument \S\ref{sec:jd} tests.

\paragraph{Why three baselines.}\label{sec:threeaxes} The viewer and item
baselines are historical per-viewer and per-item averages. The
request baseline is read from the live candidate pool of one
(viewer,~session), so it tracks what that pool contains. A viewer's
gap between connected and unconnected content, for instance, is invisible to
both averages. Yet in a connected-heavy request it raises the level of a batch
that pools both kinds, against which an unconnected candidate is read.

The request axis carries two effects. \emph{Batch-strength} bias arises
\emph{across} requests: the same score is exceptional in a weak request
and ordinary in a strong one. \emph{Composition} bias arises \emph{within} a
request, from its mix of candidates, as in the connected example. A statistic
of the request's own predictions tracks both. Targeting
different quantities, the three terms are additive and coexist on one signal.
They overlap rather than conflict, and joint tuning resolves it.

\subsection{Three properties of the operator}
\label{sec:properties}

The results below concern the score we construct and make no observational
baseline causal. Parts of them are known in other literatures.
Theorem~\ref{thm:scale} is an instance of maximal
invariance~\cite{lehmann2005testing}, Theorem~\ref{thm:gate}(i) is the binary
case of entropy-regularised prediction~\cite{blondel2020learning}, and
$e\,\sigma(ke) = \mathrm{SiLU}(ke)/k$ is a rescaled SiLU
unit~\cite{elfwing2018sigmoid}. Proofs are in Appendix~\ref{app:proofs}. The $O(1/k)$
bound of Theorem~\ref{thm:gate}(ii), with constant $\approx 0.278$, appears
in~\cite{zhang2024deep}. What is ours is the reading of each for ranking.

\phantomsection\label{sec:thm-scale}
\begin{theorem}[The logit gap is a maximal invariant]
\label{thm:scale}
Let $T_a(x) = \sigma\big(\operatorname{logit}(x) + a\big)$ shift the log-odds
by $a \in \mathbb{R}$. A comparison $D : (0,1)^2 \to \mathbb{R}$ satisfies
\[
D\big(T_a(p), T_a(b)\big) = D(p, b) \quad \text{for every } a \in \mathbb{R}
\]
if and only if
\[
D(p, b) = h\big(\operatorname{logit}(p) - \operatorname{logit}(b)\big)
\]
for some $h : \mathbb{R} \to \mathbb{R}$.
\end{theorem}

\paragraph{Discussion.} The logit gap carries everything a shift-invariant
comparison can use, for any base rate, so gates on probabilities compare in
log-odds; the ratio $p/b$ is not invariant. Invariance to a common rescaling or
translation gives $\log \hat y - \log b$ or $\hat y - b$ for regression signals. At
$k \to \infty$ every scale collapses to $\mathbf{1}\{d > 0\}$, so the scale of
$d$ matters only at finite $k$. A hinge's raw-unit magnitude $m = e$ keeps the
two terms of Eq.~\ref{eq:twoterm} commensurate.

\phantomsection\label{sec:thm-gate}
\begin{theorem}[The gate is a regularised threshold, the hinge its limit]
\label{thm:gate}
\textup{(i)} For an advantage $d \in \mathbb{R}$ and a temperature $\tau > 0$,
the objective
\[
J_\tau(g; d) = g\,d - \tau\, D_{\mathrm{KL}}\big(\mathrm{Bern}(g) \,\|\, \mathrm{Bern}(1/2)\big),
\qquad 0 < g < 1,
\]
is strictly concave in $g$, with unique maximiser $g^\star = \sigma(d/\tau)$.
\par\noindent\textup{(ii)} For fixed $m$ and every comparison $d \neq 0$,
\[
m\,\sigma(k d) \;\xrightarrow[k \to \infty]{}\; m \cdot \mathbf{1}\{d > 0\}.
\]
At $m = d = e$ the gate $H_k(e) = e\,\sigma(ke)$ tends uniformly to the hinge
$H_\infty(e) = \max(0, e)$:
\[
\sup_{e \in \mathbb{R}} \big| H_k(e) - H_\infty(e) \big|
= \frac{\mathrm{W}(1/\mathrm{e})}{k} < \frac{0.279}{k},
\]
with $\mathrm{W}$ the Lambert $W$ function and $\mathrm{e}$ Euler's number.
\par\noindent\textup{(iii)} At every $k$ the promotion and demotion forms
partition the gap,
\[
e\,\sigma(ke) + e\,\big(1 - \sigma(ke)\big) = e,
\]
mirroring $\max(0, e) + \min(0, e) = e$.
\end{theorem}

\paragraph{Discussion.} Part~(i) makes $k = 1/\tau$ an inverse temperature.
By part~(ii), $m$ sets the limit: the hinge at $m = e$, a flat $-C$ demotion
at constant $m$, and a jump of size $m$ at the baseline for the deployed gates. So
$m$ and $k$ are separate choices. A same-magnitude gate and hinge differ only
within a few multiples of $1/k$ of the baseline. A logit gap therefore suits
$k = 1$, while a raw probability gap, often below $10^{-3}$, needs a very
large $k$.

\phantomsection\label{sec:thm-two}
Let $P$ and $Q$ list one signal's $p(c)$ and $\Delta(c)$ over the candidates
$\mathcal{C}$ of a set of requests. With the rest of the VM fixed, the raw
weight alone and both weights reach
\[
\mathcal{S}_{\mathrm{raw}} = \{\alpha P : \alpha \in \mathbb{R}\},
\qquad
\mathcal{S}_{\mathrm{both}} = \{\alpha P + \beta Q : \alpha, \beta \in \mathbb{R}\}.
\]
Ranking within a request is unchanged by adding any vector that is constant
within each request; let $\mathcal{K}$ be the subspace of such vectors.

\begin{theorem}[The two-term form is strictly more expressive]
\label{thm:two}
\[
\mathcal{S}_{\mathrm{raw}} + \mathcal{K} \subseteq \mathcal{S}_{\mathrm{both}} + \mathcal{K},
\quad\text{with equality iff } Q \in \operatorname{span}(P) + \mathcal{K}.
\]
Hence, with the rest of the VM fixed and weights of either sign, the two-term
form reaches score vectors, and generically rankings, that no reweighting of
the raw term reaches.
\end{theorem}

\paragraph{Discussion.} Viewer and request baselines are constant within a
request, so a linear $\Delta = p - b$ applied to every candidate ranks exactly
like a larger raw weight;
only the gate or hinge nonlinearity takes $Q$ outside
$\operatorname{span}(P) + \mathcal{K}$. On the item axis a linear $\Delta$ gives MACR's
raw-minus-weighted-baseline inference and PPAC's global-popularity
term~\cite{wei2021model,ning2024debiasing}. Because $\sigma$ and the hinge are
nonlinear, $Q \notin \operatorname{span}(P) + \mathcal{K}$ generically, which is
checkable on logged candidates. The theorem says raising the raw
weight is not an equivalent parameterisation, not that a fixed two-term setting
is better or the production point a good one. The direct
$\textsc{both} - \textsc{raw}$ and $\textsc{both} - \textsc{incremental}$
contrasts of \S\ref{sec:jd} test that.

\section{The Design Procedure}
\label{sec:approach}

Eq.~\ref{eq:operator} says what is expressible. The steps below, summarised in
Algorithm~\ref{alg:procedure}, choose a point in it for a new signal, and each
names the axes it applies to. On the request axis the shape measurement of
Step~3 replaces Step~1 but decides only how to build the term.

\begin{algorithm}[t]
\caption{Designing a term for a signal's prediction $p$.}
\label{alg:procedure}
\footnotesize
\begin{algorithmic}[1]
\State \textbf{Step 1} (viewer, item; \S\ref{sec:diagnostic}): quantile $r$ and
  Var\% of $p$ vs.\ each baseline (strong: $r \ge 0.5$, Var\% $\ge 20\%$; weak:
  $r < 0.2$, Var\% $< 10\%$)
\State \quad a weak $r$ on a sparse feature (zero share $> 60\%$) is
  unresolved: lower its priority, do not discard it
\State \textbf{Step 2} (viewer, item; \S\ref{sec:estimate}): estimate $b$ as a
  rate or a format-and-size bucket; choose its window and clamp
\State \textbf{Step 3} (request; \S\ref{sec:tree}): measure the within-request shape
\State \quad right-skewed: promote $\sim$p75, linear hinge; demote p10/p25, sigmoid gate
\State \quad symmetric: promote p75/p90, linear hinge; demote p10/p25, linear hinge
\State \textbf{Step 4} (item; \S\ref{sec:stage}): popularity bias $\to$ VM term;
  media-type bias $\to$ debiased label
\State \textbf{Step 5} (all; \S\ref{sec:cost}): check the cost; add $\Delta$ at one
  point $(b, d, m, k, \text{shift})$ to production and test it
\end{algorithmic}
\end{algorithm}

\begin{table}[t]
\caption{[Offline] Viewer-axis diagnostics. ``Zero'' is the share of impressions with a zero feature, on which the clamp binds. Var\% is computed against the prediction. Feed time spent has no zero share because its feature is in seconds.}
\label{tab:diaguser}
\footnotesize
\begin{tabular}{@{}lrrrr@{}}
\toprule
signal & quantile $r$ & zero \% & Var\% & $n$ \\
\midrule
like & $+0.687$ & 6 & 22.0 & 2.90M \\
comment\_tap & $+0.652$ & 15 & 17.3 & 2.89M \\
reshare\_select\_recipient & $+0.593$ & 40 & 45.4 & 2.67M \\
reshare\_button\_tap & $+0.538$ & 21 & 16.6 & 2.89M \\
profile\_tap & $+0.525$ & 8 & 9.3 & 2.90M \\
save & $+0.512$ & 43 & 17.4 & 2.67M \\
comment & $+0.463$ & 64 & 26.5 & 2.89M \\
feed\_time\_spent & $+0.438$ & \textemdash & 48.3 & 1.56M \\
external\_share & $+0.245$ & 89 & 2.7 & 2.67M \\
follow & $+0.148$ & 71 & 9.8 & 2.67M
\\
\bottomrule
\end{tabular}
\end{table}

\subsection{Step 1 (viewer, item): is the split needed?}
\label{sec:diagnostic}

The \emph{quantile Pearson} $r$ correlates the decile-bucket indices of a
signal's prediction and its baseline feature. This rank statistic is needed
because some baseline rate features carry non-physical values above $1.0$.
Raw Pearson is not interchangeable, as our bands
(Algorithm~\ref{alg:procedure}) are calibrated on the rank form. A high $r$ means the prediction tracks the
baseline rather than anything candidate-specific. A variance ratio,
$\text{Var\%} = \operatorname{Var}(b) / \operatorname{Var}(y)$ for $y$ the
label or the prediction, measures the dependence's \emph{magnitude}, $r$
its \emph{shape}. Their agreement makes a reading reliable. Both
are computed on one week of candidate-level ranking logs, a $1\%$ block-level
sample. Quantile correlations move by about $\pm 0.01$ between runs.

Var\% needs proxy and target to share construct and units. This rules it out
for implicit signals labelled on a different surface from the feature and for
binary labels against a feature in seconds. We report it on the viewer axis
only: on the item axis three signals exceed $100\%$, a precondition violation,
not a property of the signals.

\paragraph{The base rate shapes the reading.} Frequent actions read highest:
like and comment tap, the commonest, lead Table~\ref{tab:diaguser}, and follow,
the rarest, reads lowest. A frequent action becomes a reflex,
which is what a viewer baseline measures; a rare one leaves too little history
to habituate. Part of the pattern is also mechanical: a rare action leaves the
viewer feature mostly at zero, which caps how high a rank-based $r$ can go.
External share reads low although it is not the rarest action, consistent with
its higher friction and a sparse viewer feature ($\extShareZero\%$ zero).
The base rate therefore also decides whether the baseline is resolvable.
Every logit-scale term clamps its baseline away from $0$, and on a sparse
feature the clamp binds for most impressions, where the feature says little
about the viewer. The share of
impressions where it binds, the zero column of Table~\ref{tab:diaguser},
exceeds $60\%$ for three of nine viewer signals. A weak $r$ on such a signal is
\emph{unresolved}, not \emph{weak}.

\subsection{Step 2 (viewer, item): estimating the baseline}
\label{sec:estimate}

On the viewer axis, implicit signals have no rate of their own and share the
viewer's feed time spent as a proxy. On the item axis, a popularity rate needs
a minimum-impression guard and is partly a product of the ranker. It is thus a
prior, not a causal quantity: one more reason the raw term stays. A
\emph{media-type} baseline comes from the item's format-and-size bucket (video
duration, carousel card count, photo OCR word count). Estimated across items, it is
dense where a popularity rate is sparse, which helps \S\ref{sec:stage} route it
through the label. A historical rate needs an aggregation window, trading noise
against staleness. A logit-scale term needs a clamp, setting its behaviour
where the feature is empty. Base rate and zero share narrow both, and the final
choice is made online.

\begin{figure}[t]
\Description{Box plots of within-request prediction shapes, and a bar chart of the coefficient of variation of candidate baseline statistics.}
\centering
\begin{tikzpicture}[x=1cm, y=1cm]

\def\W{3.9}
\def\rowsep{0.37}

\newcommand{\shapebox}[9]{%
  \pgfmathsetmacro{\py}{-#1*\rowsep}%
  \ifthenelse{\equal{#3}{skewed}}%
    {\colorlet{bc}{orange!75!black}}{\colorlet{bc}{teal!65!black}}%
  \node[font=\scriptsize, anchor=east] at (-0.10,\py) {#2};
  \draw[bc!55, line width=0.5pt] ({#4*\W},\py) -- ({#8*\W},\py);
  \fill[bc!22] ({#5*\W},\py-0.115) rectangle ({#7*\W},\py+0.115);
  \draw[bc, line width=0.5pt]
      ({#5*\W},\py-0.115) rectangle ({#7*\W},\py+0.115);
  \draw[bc, line width=1.1pt] ({#6*\W},\py-0.125) -- ({#6*\W},\py+0.125);
  \node[font=\scriptsize, anchor=west, black!70] at (\W+0.14,\py) {#9};
}

\node[font=\scriptsize\bfseries, anchor=south west] at (-1.55,0.46)
  {(a) within-request shape};
\draw[black!22, dashed] (0.5*\W,0.30) -- (0.5*\W,-5*\rowsep-0.18);
\node[font=\scriptsize, black!65, anchor=south] at (0.5*\W,0.30) {mid};
\node[font=\scriptsize, black!70, anchor=south] at (\W+0.30,0.30) {skew};

\input{tables/fig_shape_data}

\node[font=\scriptsize, anchor=north, black!70] at (\W/2,-5*\rowsep-0.20)
  {normalised p10\,--\,p90 range \quad (box: p25\,--\,p75, rule: p50)};

\begin{scope}[shift={(0,-4.80)}]
  \def\CW{4.6}\def\CH{1.25}\def\cvmax{50}
  \def\step{0.72}\def\bw{0.20}

  \newcommand{\cvbar}[4]{%
    \pgfmathsetmacro{\bx}{#1*\step + 0.36}%
    \pgfmathsetmacro{\hs}{#3/\cvmax*\CH}%
    \pgfmathsetmacro{\hy}{#4/\cvmax*\CH}%
    \fill[orange!75!black] (\bx-\bw,0) rectangle (\bx,\hs);
    \fill[teal!60!black]   (\bx,0) rectangle (\bx+\bw,\hy);
    \node[font=\scriptsize, anchor=south, orange!55!black]
      at (\bx-\bw/2,\hs+0.02) {#3};
    \node[font=\scriptsize, anchor=north, black!60] at (\bx,-0.06) {#2};
  }

  \node[font=\scriptsize\bfseries, anchor=south west] at (-1.55,\CH+0.42)
    {(b) stability of the baseline statistic};
  \draw[black!65] (0,0) -- (\CW,0);
  \foreach \v in {10,20,30,40,50} {
    \pgfmathsetmacro{\vy}{\v/\cvmax*\CH}
    \draw[black!12] (0,\vy) -- (\CW,\vy);
    \node[font=\scriptsize, black!65, anchor=east] at (-0.04,\vy) {\v};
  }
  \node[font=\scriptsize, black!70, rotate=90, anchor=south]
    at (-0.44,\CH/2) {CV across requests (\%)};

  \input{tables/fig_cv_data}

  \draw[black!70, dashed, line width=0.4pt]
    (2*\step+0.36,-0.02) -- (2*\step+0.36,\CH+0.22);
  \node[font=\scriptsize\bfseries, black!65, anchor=south]
    at (2*\step+0.36,\CH+0.22) {chosen};

  \begin{scope}[shift={(3.05,\CH+0.40)}]
    \fill[orange!75!black] (0,0) rectangle (0.22,0.14);
    \node[font=\scriptsize, anchor=west] at (0.27,0.07) {right-skewed};
    \fill[teal!60!black] (1.75,0) rectangle (1.97,0.14);
    \node[font=\scriptsize, anchor=west] at (2.02,0.07) {symmetric};
  \end{scope}
\end{scope}

\end{tikzpicture}
\caption{\capBatch}
\label{fig:batch}
\end{figure}

\subsection{Step 3 (request): shape selects the statistic}
\label{sec:tree}

Low-rate binary signals are right-skewed within a request: bounded below by
zero with mean at the base rate, they can spread only upward. Log time spent
and $p_{10s}$, the predicted probability that time spent on the item exceeds
10\,s, are near-symmetric (Figure~\ref{fig:batch}(a)), as are
bucketised labels by construction. \emph{The floor must be de-skewed}: a linear
rule gives a skewed signal's compressed bottom mass nearly identical scores,
while a logit gap separates it. A finite-$k$ term on a probability needs this
on every axis. A promotion hinge is exempt: it is zero below $b$, where the
floor is compressed (\S\ref{sec:thm-scale}).
\emph{The statistic must be stable across requests}, a concern only on the
request axis, where the baseline comes from a small, fresh sample (Figure~\ref{fig:batch}(b)). For a symmetric signal any statistic
will do. For a skewed one the high percentiles are noisy, the tail drags the
mean to around percentile \meanAtPct{}, and the median sits on the floor. That
leaves p75 as the steadiest off-floor statistic. Adding direction gives the
rule in Algorithm~\ref{alg:procedure} (lines 5--6). Only one case needs the sigmoid: demoting a right-skewed signal, whose
floor-compressed candidates separate only in log-odds.

\subsection{Step 4 (item): VM term or debiased label?}
\label{sec:stage}

Instead of a VM term, a baseline can be removed from the label: a candidate is
positive only if it out-performs its own bucket on the bias dimension. A VM
term is a computationally cheap, tunable, reversible configuration change but sees only the
score and one baseline feature. A debiased label adds a new predictor and its
serving cost, but the model learns the correction and trains on the debiased
target. The choice arises only on the item axis and follows the bias present
(Table~\ref{tab:diagitem}). \emph{Popularity} bias gets a VM term: its baseline
is one item-level statistic of the same signal, so a post-hoc gap captures
it. The label is undistorted, and for a binary signal a debiased label is
ill-defined per impression. \emph{Media-type} bias gets a label: the
distortion is in the label, and a high-cardinality non-linear correction in the
VM would entangle with every other term. On a continuous signal, ``out-performs
its own size class'' is well defined. The bucketised labels
(\S\ref{sec:related}) take this branch, the model-side exception to the
scoring-time terms. The two are not exclusive: a VM term can precede the label.

\begin{table}[t]
\caption{[Offline] Item-axis diagnostics for popularity (left) and video duration (right). A further \nDurHidden{} duration signals, all with $|r| < 0.09$, are omitted.}
\label{tab:diagitem}
\footnotesize
\begin{tabular}{@{}lr@{\hskip 1.0em}lr@{}}
\toprule
popularity & $r$ & duration & $r$ \\
\midrule
profile\_tap & $+0.599$ & vc\_ratio\_80 & $-0.418$ \\
reshare\_button\_tap & $+0.521$ & log\_time & $+0.216$ \\
like & $+0.456$ & comment\_ts & $+0.121$ \\
follow & $+0.407$ & ts\_p75\_carousel\_v2 & $+0.105$ \\
feed\_time\_spent & $+0.400$ & reelsiv\_chain\_view\_time & $+0.084$ \\
save & $+0.352$ & ts\_p75\_video\_v2 & $-0.017$
\\
\bottomrule
\end{tabular}
\end{table}

\begin{table*}[t]
\caption{[Online] The six studied terms, all live. \emph{Form}: the deployed incremental term $\Delta$, unweighted (Eq.~\ref{eq:operator}; for bucketTS-p75, debiased labels), coloured by role: \swatch{cPred}~prediction, \swatch{cBase}~baseline, \swatch{cGate}~the comparison that gates it. Elsewhere a term is named signal\,$\cdot$\,axis; TS = time spent, MR/TR = MainRank/TopRank. $^\dagger$Shipped in a bundle; the number is the labels' own for bucketTS-p75 and the bundle's for the other two. \sig{}~as in Table~\ref{tab:main}.}
\label{tab:deployment}
\footnotesize
\setlength{\tabcolsep}{4pt}
\begin{tabular}{@{}ll>{\raggedright\arraybackslash}p{0.625\textwidth}>{\raggedright\arraybackslash}p{0.15\textwidth}@{}}
\toprule
term & axis & form & approved on \\
\midrule
\mbox{bucketTS-p75}$^\dagger$ & item
 & photo: $\Delta = \roleP{\hat P}\big(t \roleG{\,>\,} \roleB{b_{\text{words}}}\big)$; \quad carousel: $\Delta = \roleP{\hat P}\big(t \roleG{\,>\,} \roleB{b_{\text{cards}}}\big)$; \quad video: $\Delta = \roleP{\hat P}\big(t \roleG{\,>\,} \roleB{b_{\text{duration}}}\big)$\newline
   {\scriptsize \roleP{$\hat P$}: predicted probability; $t$: time spent; \roleB{$b$}: p75 of $t$ in the post's bucket of OCR word count, number of cards or duration.}
 & $+0.447\%$\sig{} feed time spent \\
\mbox{logTS-MR} & item
 & $\Delta = \roleP{p}\;\mathopen{\roleG{\sigma\bigl(}}\roleP{\hat t} - \roleB{b}\mathclose{\roleG{\bigr)}}$\newline
   {\scriptsize \roleP{$p$}: predicted log time spent; \roleP{$\hat t$}: predicted time spent; \roleB{$b$}: the post's mean time spent per view.}
 & $+0.665\%$\sig{} impressions \\
seeMore & request
 & $\Delta = \mathopen{\roleG{\max\bigl(0,}}\;\roleP{p} - \roleB{b}\mathclose{\roleG{\bigr)}}$\newline
   {\scriptsize \roleP{$p$}: P(viewer taps ``See more''); \roleB{$b$}: mean of $p$ over the request's candidates; $p$ and $b$ for unconnected candidates only.}
 & $+0.63\%$\sig{} unconnected post interested ratio \\
\mbox{logTS-TR} & item
 & $\Delta = \roleP{\hat t}\;\mathopen{\roleG{\sigma\bigl(}}\roleP{\hat t} - \roleB{b}\mathclose{\roleG{\bigr)}}$\newline
   {\scriptsize \roleP{$\hat t$}: predicted time spent; \roleB{$b$}: the post's mean time spent per view.}
 & $+0.066\%$\sig{} sessions \\
p10s$^\dagger$ & request
 & $\Delta = \mathopen{\roleG{\max\bigl(0,}}\;\roleP{p} - \roleB{b}\mathclose{\roleG{\bigr)}}$\newline
   {\scriptsize \roleP{$p$}: P(time spent on the post $>$ 10\,s); \roleB{$b$}: p75 of $p$ over the request's candidates.}
 & $+0.034\%$\sig{} sessions \\
extShare$^\dagger$ & viewer
 & $\Delta = \roleP{p}\;\mathopen{\roleG{\sigma\bigl(}}\operatorname{logit}\roleP{p} - \operatorname{logit}\roleB{b}\mathclose{\roleG{\bigr)}}$\newline
   {\scriptsize \roleP{$p$}: P(viewer shares the post outside Instagram); \roleB{$b$}: the viewer's own external-share rate.}
 & $+0.030\%$\sig{} sessions \\
\bottomrule
\end{tabular}
\end{table*}

\begin{table}[t]
\caption{[Online] Four-arm ablation. Columns 2--4: each arm relative to \textsc{none}; columns 5--6: the direct contrasts \textsc{both}$\,-\,$\textsc{raw} and \textsc{both}$\,-\,$\textsc{incr}. \sig{}~marks $|\delta|$ exceeding the confidence half-width. \textsc{incr} = \textsc{incremental}; JD = joint dominance. Indented rows are the judged (term, metric) pairs of \S\ref{sec:protocol}.}
\label{tab:main}
\footnotesize
\setlength{\tabcolsep}{2pt}
\begin{tabular}{@{}lrrrrrc@{}}
\toprule
 & \multicolumn{3}{c}{vs.\ \textsc{none}} & \multicolumn{2}{c}{\textsc{both}\,$-$} & \\
\cmidrule(lr){2-4}\cmidrule(lr){5-6}
metric & \textsc{raw} & \textsc{incr} & \textsc{both} & \textsc{raw} & \textsc{incr} & JD \\
\midrule
\rlap{\pTenS} & & & & & & \\
\quad sessions & $+0.029\%$\sig & $+0.035\%$\sig & $+0.047\%$\sig & $+0.018\%$ & $+0.012\%$ & \yes \\
\quad feed time spent & $+1.286\%$\sig & $+0.695\%$\sig & $+1.650\%$\sig & $+0.359\%$\sig & $+0.948\%$\sig & \yes \\
\rlap{\logTSmr} & & & & & & \\
\quad sessions & $+0.054\%$\sig & $+0.042\%$\sig & $+0.083\%$\sig & $+0.029\%$\sig & $+0.040\%$\sig & \yes \\
\quad feed time spent & $+0.203\%$\sig & $+0.196\%$\sig & $+0.353\%$\sig & $+0.150\%$\sig & $+0.157\%$\sig & \yes \\
\quad impressions & $-0.103\%$\sig & $+0.773\%$\sig & $+0.639\%$\sig & $+0.743\%$\sig & $-0.132\%$\sig & \no \\
\rlap{\bucketTS} & & & & & & \\
\quad sessions & $+0.045\%$\sig & $+0.001\%$ & $+0.053\%$\sig & $+0.007\%$ & $+0.052\%$\sig & \yes \\
\quad feed time spent & $+0.174\%$\sig & $+0.182\%$\sig & $+0.318\%$\sig & $+0.143\%$\sig & $+0.135\%$\sig & \yes \\
\rlap{\logTStr} & & & & & & \\
\quad sessions & $-0.008\%$ & $+0.007\%$ & $+0.001\%$ & $+0.009\%$ & $-0.007\%$ & \no \\
\quad feed time spent & $-0.002\%$ & $+0.200\%$\sig & $+0.216\%$\sig & $+0.218\%$\sig & $+0.017\%$ & \yes \\
\quad impressions & $-0.036\%$ & $-0.041\%$\sig & $-0.022\%$ & $+0.015\%$ & $+0.019\%$ & \yes \\
\rlap{\extshareUser} & & & & & & \\
\quad sessions & $+0.047\%$\sig & $+0.067\%$\sig & $+0.088\%$\sig & $+0.041\%$\sig & $+0.021\%$\sig & \yes \\
\quad ext.\ share & $+31.32\%$\sig & $+43.63\%$\sig & $+58.97\%$\sig & $+21.05\%$\sig & $+10.68\%$\sig & \yes \\
\rlap{\seemore} & & & & & & \\
\quad \begin{tabular}[t]{@{}l@{}}unconnected post\\\ \ interested ratio\end{tabular} & $+1.241\%$\sig & $+1.495\%$\sig & $+1.929\%$\sig & $+0.681\%$\sig & $+0.427\%$\sig & \yes
\\
\bottomrule
\end{tabular}
\end{table}
\subsection{Step 5 (all axes): cost, and what to ship}
\label{sec:cost}

Viewer and item baselines are precomputed features with no scoring cost. A
request term pays per request, $O(N)$ in the pool size for a mean and
$O(N\log N)$ for a percentile, computationally cheap for one term but not for every signal.
No scoring-time term adds model inference: each reads predictions and features the ranker
already computes. Measured latency costs, app-start latency included, are a
few milliseconds at most.
The procedure ships one point after one test (Algorithm~\ref{alg:procedure},
line 8). One arm suffices
because the raw term is not in question at ship time. Other points are swept
only for signals worth the traffic. The four-arm design of \S\ref{sec:setup} is
a separate, retrospective study of whether the two-term form was right at all.
No launch answers that question, because none runs the arm without the signal.

\section{Deployment}
\label{sec:deployment}

The baseline--incremental decomposition has been shipping to Instagram Feed
production for most of a year; the six studied terms launched over many months
and are all live. Table~\ref{tab:deployment} gives each term's place in the
design space of Eq.~\ref{eq:operator} and its approval metric. Each term serves all Feed traffic in its ranking stage. All
three axes appear, with gates, hinges and debiased labels. The two feed
time-spent gates read the same prediction against the same item baseline but
differ in magnitude as well as in their VM; they are two related points.
Table~\ref{tab:main} shows that they behave differently.

The launch bar is not uniform: a term is approved on the metric it was built to
move, and the protocol of \S\ref{sec:setup} relies on this. Three of the six
(two within a bundle) cleared a significant, replicated foreground-session
gain. The see-more term, which predicts whether the viewer taps ``See more''
to ask for more posts like a recommended one, shipped on $+0.63\%$ unconnected
post interested ratio
at the cost of $-0.118\%$ feed time spent. The MainRank gate left sessions
flat. It was approved on the rebalancing it produced: $+0.665\%$ impressions
and $+2.21\%$ impressions outside the video viewer, as photos, carousels and
shorter content gained position. The three bucketised labels,
ablated together, shipped within a larger MainRank restructuring. On their own they contributed
$+0.447\%$\sig{} feed time spent and $+0.002\%$ sessions. The table is
therefore not a list of independent session wins.

Coverage, not efficacy, is the binding constraint. Of the Feed VM
signals, fewer than $20\%$ have a matching viewer-side rate feature, about six
have an item-side one, and about five carry a request-level term. Most of the
VM cannot be debiased today for want of a baseline, so the computationally cheap diagnostic of
\S\ref{sec:diagnostic} and its preconditions matter.

\section{The Ablation}
\label{sec:setup}

Each term is ablated in a four-arm online experiment on live traffic. The
evaluation is online because a VM term's effect depends on the other
terms and on how users respond, which offline replay cannot measure. Users are
randomised, each experiment in its own universe (two share one, with disjoint
arms), and each arm gets $3\%$ of traffic. Each
delta is computed by one method on one seven-day period, relative to
\textsc{none}, with $95\%$ intervals. It was the only
week all six experiments were live and ramped; each ran twelve to
seventeen days. Each
experiment crosses the two weights $(w^{\text{r}}, w^{\text{i}})$ of Eq.~\ref{eq:twoterm},
with prod the production weight: \textsc{none} $(0, 0)$, \textsc{raw}
$(\text{prod}, 0)$, \textsc{incremental} $(0, \text{prod})$ and \textsc{both}
$(\text{prod}, \text{prod})$. \textsc{none}, the common origin, is the control
for every reported delta. \textsc{both} is production for every reported term
except that \textsc{raw} was production for \extshareUser{} during the
measurement window. Its incremental term launched the day after the window
closed.
As comparators, \textsc{raw} is the status quo without debiasing,
\textsc{incremental} the removal design common in the literature
(\S\ref{sec:related}), and \textsc{both} ours. $\textsc{both} - \textsc{raw}$
adds the incremental term at its production weight to a VM that already
carries the raw term, and $\textsc{both} - \textsc{incremental}$ the converse.
Neither single-term arm was re-tuned. This favours \textsc{both} over
\textsc{incremental}, whose weight was set beside the raw term. The design
thus tests whether either term is redundant at the deployed point, not whether a
re-tuned single-term VM could match it. Six experiments were run, one for each
term live in production and designed through the procedure of
\S\ref{sec:approach}.

\paragraph{Evaluation protocol.}
\label{sec:protocol}
Judging terms on a full metric panel invites false positives. The six
experiments give $270$ cells, and $95\%$ intervals alone would flag about
fourteen of them by chance. We therefore fixed, per term, the metrics it would be
judged on: the criteria it would have been launched on,
assigned as in Table~\ref{tab:main}. No significant cell outside that
assignment is used as evidence.
For five of the six terms the primary metric is foreground sessions. It is the
best available proxy for daily actives and, because a session is a separate act
of returning, the hardest to inflate by re-ordering. At Instagram's scale, even
a small relative gain in sessions is practically meaningful. Each term is also judged on
its \emph{mechanism} metric, the behaviour it is built to change. A term whose
mechanism metric does not move has not done its job. The two feed time-spent
gates are judged on feed time spent and impressions, since shifting attention
across candidates of different length can cash out as time spent or as
impressions. The sixth term, \seemore{}, is a content-quality signal whose launch
criterion was the unconnected post interested ratio. It is judged on that
alone, hence \nPairs{} pairs rather than fourteen. Both of its session contrasts
are negative and not significant.

\paragraph{Joint dominance.} A (term, metric) pair is \emph{jointly dominant}
when the point estimates satisfy $\textsc{both} > \textsc{raw}$ and
$\textsc{both} > \textsc{incremental}$. Each is read from the direct arm-to-arm
contrast, never from a difference of two deltas against
\textsc{none}. We use point estimates because sessions' half-width is about
$\pm 0.020\%$ against direct session contrasts of at most about $0.05\%$, so
requiring significance on both contrasts would set aside most session-judged
pairs. We report significance throughout,
counting significantly positive and negative contrasts separately, so a
stricter reader can re-judge.

\paragraph{Results.}
\label{sec:jd}
Of the \nPairs{} (term, metric) pairs, \nJD{} are jointly dominant
(Table~\ref{tab:main}). Each of the \nCands{} terms is jointly dominant on at
least one of its metrics, and \nJDCandsAll{} of \nCands{} on all of them. Of the
\nLegs{} direct contrasts, \nLegsPos{} are significantly positive and \nLegsNeg{}
is significantly negative. The positive side includes $+21.05\%$ on external
shares and $+0.359\%$ on feed time spent, against arms that already carry one of
the two terms at production weight. Power differs by metric family. On the
mechanism metrics, \nLegsMechPos{} of \nLegsMech{} contrasts are significantly
positive. On sessions, \nLegsSesPos{} of \nLegsSes{} are, and \nSesJD{} of the
\nSesPairs{} session-judged terms are jointly dominant on sessions. Session
verdicts therefore rest more often on non-significant contrasts.

\textsc{raw} and \textsc{incremental} are \emph{restrictions} of \textsc{both},
each with one weight set to zero, so \textsc{both}'s reachable set contains them
by construction. The ablation tests whether the extra degree of freedom is
\emph{used}, that is, whether the interior point, with both weights non-zero,
beats both boundary points, each with one weight zero. On \nJD{} of \nPairs{}
pairs it does: at the deployed operating point neither term is redundant. The \nCands{} terms are a case series, not a
sample, so we draw no cross-term rank correlations. Every load-bearing claim
sits within one term.

\paragraph{Why not simply raise the raw weight?}
\label{sec:whynotraw}
\looseness=-1
The objection is that \textsc{both} puts more total weight on the signal than
\textsc{raw}, so a re-tuned $w^{\text{r}}$ alone might reach the same place. The
ablation has four points, not a weight sweep. Theorem~\ref{thm:two} addresses the
formal half.
Whether the production point is good rests on how the VM is maintained. Its
weights are the standing output of a continuous tuning loop. Each term in
Table~\ref{tab:deployment} passed launch review against the full metric panel and its
guardrails, and has since been re-examined by boost and ablation experiments.
Raising one weight is that loop's simplest and most frequent experiment.
\textsc{raw} at production weight is therefore close to the best the
single-term form has reached under the same bar. A lift of $w^{\text{r}}$ alone
that bought what \textsc{both} buys would likely already have shipped. The argument
is \emph{local}: it says the production point is hard to improve by moving one
weight, not that it is globally best. It holds more weakly for $w^{\text{i}}$,
new on most of these signals and far less tuned, so if anything the
ablation understates the two-term form. Before the four-arm design, earlier
viewer-axis experiments ran the sweep this objection asks for: one signal at a
time, each debias experiment paired with a matched weight-only one. The
debiased arm consistently converted to sessions where the weight-only arm did
not.

\paragraph{The two misses.}
\label{sec:failures}
\looseness=-1
Two of the \nPairs{} pairs are not jointly dominant. \logTStr{} on sessions
misses by $\textsc{both}-\textsc{incremental} = -0.007\%$, well inside its
half-width, and is jointly dominant on both its mechanism metrics.
\logTSmr{} on impressions gives the study's only significantly negative
contrast, $\textsc{both}-\textsc{incremental} = -0.132\%$\sig{}. On impressions
the raw term alone is negative ($-0.103\%$\sig{}), while the incremental term
alone gives the largest impression effect in the table ($+0.773\%$\sig{}). On
feed time spent both single terms are positive and close, and \textsc{both}
beats each. The raw term buys time spent; the incremental term changes what is
shown, toward shorter and non-video content. This is a real trade-off between
two objectives, which separately weighted terms expose and a single coefficient
would hide. It is also the trade-off the term launched on: the MainRank gate
was approved on impressions, its TopRank counterpart on sessions
(Table~\ref{tab:deployment}).

\section{Lessons from Deployment}
\label{sec:discussion}

\begin{figure}[t]
\Description{Offline correlation of every audited signal on its bias axis, with tested signals marked.}
\centering
\begin{tikzpicture}[x=1cm, y=1cm]

\def\xlo{-0.50}\def\xhi{0.75}\def\W{6.6}
\newcommand{\xc}[1]{{(#1 - \xlo) / (\xhi - \xlo) * \W}}
\def\rowsep{0.92}

\newcommand{\signalpt}[4]{%
  \pgfmathsetmacro{\px}{\xc{#2}}%
  \pgfmathsetmacro{\py}{-#1*\rowsep}%
  \ifthenelse{\equal{#3}{none}}{%
    \fill[black!28] (\px,\py) circle (1.7pt);
  }{%
    \fill[blue!60!black] (\px,\py) circle (3.1pt);
  }%
  \ifthenelse{\equal{#3}{tested}}{%
    \node[font=\scriptsize\bfseries, blue!45!black, anchor=south]
      at (\px,\py+0.13) {#4};
  }{}%
  \ifthenelse{\equal{#3}{testedlo}}{%
    \node[font=\scriptsize\bfseries, blue!45!black, anchor=north]
      at (\px,\py-0.13) {#4};
  }{}%
}

\foreach \r/\name in {0/viewer, 1/item: popularity, 2/item: duration} {
  \draw[black!30] (0,-\r*\rowsep) -- (\W,-\r*\rowsep);
  \node[font=\scriptsize\itshape, anchor=south west, black!70]
    at (0,-\r*\rowsep+0.30) {\name};
}
\draw[black!65] (\xc{0},0.30) -- (\xc{0},-2*\rowsep-0.30);
\node[font=\scriptsize, black!65, anchor=north] at (\xc{0},-2*\rowsep-0.30) {$0$};

\foreach \t in {-0.4,-0.2,0.2,0.4,0.6} {
  \pgfmathsetmacro{\tx}{\xc{\t}}
  \draw[black!35] (\tx,-2*\rowsep-0.06) -- (\tx,-2*\rowsep-0.16);
  \node[font=\scriptsize, black!65, anchor=north] at (\tx,-2*\rowsep-0.18) {\t};
}
\node[font=\scriptsize, anchor=north] at (\W/2,-2*\rowsep-0.52)
  {quantile Pearson $r$ \; (prediction vs.\ baseline feature)};

\input{tables/fig_offon_data}

\begin{scope}[shift={(0,0.95)}]
  \fill[blue!60!black] (0.10,0) circle (2.8pt);
  \node[font=\scriptsize, anchor=west] at (0.26,0)
    {carries a term with an online result (Table~\ref{tab:main})};
  \fill[black!28] (4.85,0) circle (1.7pt);
  \node[font=\scriptsize, anchor=west] at (4.99,0) {not tested yet};
\end{scope}

\end{tikzpicture}
\caption{\capOffon}
\label{fig:offon}
\end{figure}

\looseness=-1
\vspace{2pt}\noindent\emph{L1: the procedure's first choices held.} Where Table~\ref{tab:main} can check the
procedure's rules, their first choices held. The shape rule of \S\ref{sec:tree}
puts a p75 promotion hinge on \pTenS{}, a near-symmetric signal, and the term
is jointly dominant on both its metrics. The label rule of \S\ref{sec:stage}
sends the media-type-biased time-spent family to a bucketised label.
\bucketTS{} is jointly dominant on both its metrics, with every feed time-spent
arm and contrast significant. Offline, the video label retains almost no
duration correlation ($r \approx -0.017$, against $+0.216$ for the raw log-time
prediction; Table~\ref{tab:diagitem}). The photo and carousel labels are
bucketed on dimensions we did not audit, so this check covers one label of
three. \seemore{} uses a batch \emph{mean} where the shape rule prescribes
$\sim$p75. It is jointly dominant on its mechanism metric and converts to
nothing on sessions. On the request axis, the one recipe off the rule is the
one mechanism-only result.

\looseness=-1
\vspace{2pt}\noindent\emph{L2: the VM saturates around a term.} A term's value depends on the
VM around it, which keeps moving. \logTStr{} launched on a replicated
$+0.066\%$ session win, but in this window its $\textsc{both}-\textsc{raw}$ on
sessions is $+0.009\%$ (n.s.), its session effects are flat on every arm and
its impression effects nearly so. On feed time spent the incremental term is still strong ($+0.200\%$\sig{};
\textsc{both}~$+0.216\%$\sig{}; $\textsc{both}-\textsc{raw}$~$+0.218\%$\sig{})
and the raw term is at zero ($-0.002\%$, n.s.). The mechanism persists, but its
conversion into sessions did not survive eight weeks of VM retuning since
launch. A sessions-only read would have called the term useless. The viewer
terms on like, reshare recipient and comment, all reading above $0.45$, have
never produced a stable win, which we attribute, without direct evidence, to
weights already tuned against the raw signal. None is in this ablation. Offline baseline is not
scarce on Feed; \emph{unsaturated} baseline is.

\looseness=-1
\vspace{2pt}\noindent\emph{L3: screen with the offline diagnostic, do not rank with it.} As
triage, the diagnostic did its job: it would not have screened out any tested
viewer- or item-axis term in Table~\ref{tab:main} (filled markers in Figure~\ref{fig:offon}).
One of them read ninth of ten on its axis. Offline queries are computationally cheap and experiment
slots scarce, so the offline stage should be a \emph{high-recall filter}, not a
ranking. It must not discard a signal that would have worked, and may pass more
than will succeed.

\looseness=-1
\vspace{2pt}\noindent\emph{L4: a reading needs resolvability and headroom.} The diagnostic's
two limits, which we anticipated before the ablation and
which it confirmed, should be checked in advance as preconditions alongside the correlation.
They change a signal's priority without discarding it. A sparse baseline feature gives an
\emph{unresolved} reading rather than a genuinely low one
(\S\ref{sec:diagnostic}). The one viewer-axis term in Table~\ref{tab:main},
significant on all four of its contrasts, is such a signal. Where its rate
feature is zero the clamp binds, and the gate becomes a fixed function of $p$:
extra raw weight. Above the clamp it still compares each candidate with the
viewer's own baseline. The clamp level sets the share of impressions in each
regime, a lever between raw weight and debiasing; both components remain. And a high reading
says a baseline component exists, not that the VM still has room to exploit
it, as saturation shows. Nor does the diagnostic see a term's form or the
surrounding weights: the two feed time-spent gates share one reading
($r = 0.400$) yet differ in magnitude and VM (Table~\ref{tab:deployment}) and
behave differently in MainRank and TopRank. Coverage is thin (Figure~\ref{fig:offon})
and not concentrated where readings are strongest. Further terms on the
untested high-reading viewer signals and the item axis are being tested. The
strongest item-axis reading, profile tap, is producing promising early results
in an experiment still being iterated on. Until these read out, the
quantitative offline-to-online relationship is open. A study of \nCands{}
terms cannot pin it down.

\looseness=-1
\vspace{2pt}\noindent\emph{L5: judge time spent where the ranker acts.}
\pTenS{} moves overall time spent several times less
than feed time spent. \bucketTS{}'s incremental term is significant on feed
time spent but not on overall time spent. Overall time spent aggregates surfaces the
Feed ranker does not control, diluting real effects, so every time-spent
judgement here uses feed time spent.

\section{More Deployment at Meta}
\label{sec:more}

The ablation is confined to Instagram Feed; the operator is not. To show
reach, not to evaluate, Table~\ref{tab:reels} lists a subset of its Reels
launches, made over roughly a year, with each launch's own online result, none
re-analysed. Both ranking
stages and all three axes appear (no weight-only retunes). Furthermore,
Threads has launched debiased signals in late-stage ranking, and Facebook is
testing debiased terms in ongoing experiments.

\begin{table}[!ht]
\caption{[Online] A subset of debiasing launches on the Instagram Reels tab. ESR/LSR = early-/late-stage ranking. Sessions and DAU are each launch's own online result. \sig{} is defined as in Table~\ref{tab:main}.}
\label{tab:reels}
\footnotesize
\setlength{\tabcolsep}{3.5pt}
\begin{tabular}{@{}l>{\raggedright\arraybackslash}p{0.34\columnwidth}l>{\raggedright\arraybackslash}p{0.42\columnwidth}@{}}
\toprule
stage & term & axis & approved on \\
\midrule
LSR & debiased $p_{\text{like}}$ & viewer
    & $+0.083\%$\sig{} sessions, $+0.042\%$\sig{} DAU \\
LSR & watch-time gating & item
    & $+0.049\%$\sig{} sessions, $+0.026\%$\sig{} DAU \\
ESR & watch-time gating & item
    & $+0.029\%$\sig{} sessions \\
ESR & batch $p_{\text{save}}$ champion boost & request
    & $+0.056\%$\sig{} sessions, $+0.032\%$\sig{} DAU \\
ESR & log-time gating & item
    & $+0.045\%$\sig{} sessions \\
ESR & $p_{\text{like}}$ against viewer like rate & viewer
    & $+0.075\%$\sig{} sessions, $+0.035\%$\sig{} DAU \\
LSR & content-debiased skip demotion & item
    & $+0.033\%$\sig{} sessions \\
\bottomrule
\end{tabular}
\end{table}

\section{Conclusion}

\looseness=-1
The Instagram Feed deployments argue for separating, not removing, at least at
the deployed operating point: keep the raw prediction, weight the incremental
part separately, and let the tuner choose. The three debiasing families are one operator with four degrees
of freedom, with published methods along its axes. Its baselines coexist
because two are historical averages and one is read from the live candidate
pool. In a four-arm ablation of \nCands{} deployed terms, at production
weights, on point estimates, neither term is redundant in \nJD{} of \nPairs{}
judged pairs (\nJDSig{} of \nPairs{} with both contrasts significant). The
transferable lesson: a computationally cheap offline diagnostic should screen
signals, not rank them. It says where a baseline component exists; an online
test says whether the ranker still has room to exploit it.

\bibliographystyle{ACM-Reference-Format}
\bibliography{refs}

\appendix
\section{Proofs}
\label{app:proofs}

\begin{proof}[Proof of Theorem~\ref{thm:scale}]
If $D$ is invariant, $a = -\operatorname{logit}(b)$ sends $b$ to $1/2$ and $p$
to $\sigma\big(\operatorname{logit}(p) - \operatorname{logit}(b)\big)$, so
$h(z) = D\big(\sigma(z), 1/2\big)$ works. Conversely, $T_a$ adds $a$ to both
logits and leaves their difference unchanged.
\end{proof}

\begin{proof}[Proof of Theorem~\ref{thm:gate}]
(i) $\partial J_\tau / \partial g = d - \tau \operatorname{logit}(g)$
decreases strictly and vanishes only at $g = \sigma(d/\tau)$.
(ii) The limit is that of $\sigma(kd)$. With $t = k|e|$,
$|H_k - H_\infty| = |e| / \big(1 + \mathrm{e}^{k|e|}\big)
= t / \big(k(1 + \mathrm{e}^{t})\big)$, maximal where
$\mathrm{e}^{t}(t - 1) = 1$, i.e.\ $t - 1 = \mathrm{W}(1/\mathrm{e})$, with
value $(t - 1)/k$. (iii) is immediate.
\end{proof}

\begin{proof}[Proof of Theorem~\ref{thm:two}]
The inclusion is $\beta = 0$. If $Q = \gamma P + \kappa$ with
$\kappa \in \mathcal{K}$, then
$\alpha P + \beta Q \in (\alpha + \beta\gamma) P + \mathcal{K}$; conversely,
equality gives
$Q \in \mathcal{S}_{\mathrm{both}} + \mathcal{K} = \mathcal{S}_{\mathrm{raw}} + \mathcal{K} = \operatorname{span}(P) + \mathcal{K}$.
\end{proof}

\end{document}